\documentclass[10pt, conference]{IEEEtran}

\usepackage{amsfonts}                                                          
\usepackage{epsfig}
\usepackage{amsthm}
\usepackage{graphicx}        
\usepackage{amssymb}
\usepackage{amsmath}
\usepackage{multirow}
\usepackage{mathtools}
\usepackage{epstopdf}
\usepackage{etoolbox}
\usepackage{array}
\apptocmd{\thebibliography}{\setlength{\itemsep}{4pt}}{}{}

\let\bbordermatrix\bordermatrix
\patchcmd{\bbordermatrix}{8.75}{4.75}{}{}
\patchcmd{\bbordermatrix}{\left(}{\left[}{}{}
\patchcmd{\bbordermatrix}{\right)}{\right]}{}{}
\renewcommand{\theequation}{\thesection.\arabic{equation}}
\renewcommand{\thefigure}{\thesection.\arabic{figure}}
\renewcommand{\thetable}{\thesection.\arabic{table}}

\newcommand{\sr}{\stackrel}

\newcommand{\rar}{\rightarrow}

\newcommand{\tri}{\sr{\triangle}{=}}

\newcommand{\be}{\begin{equation}}
\newcommand{\ee}{\end{equation}}
\newcommand{\bea}{\begin{eqnarray}}
\newcommand{\eea}{\end{eqnarray}}
\newcommand{\bes}{\begin{eqnarray*}}
\newcommand{\ees}{\end{eqnarray*}}
\newcommand{\bce}{\begin{center}}
\newcommand{\ece}{\end{center}}
\newcommand{\beae}{\begin{IEEEeqnarray}{rCl}}
\newcommand{\eeae}{\end{IEEEeqnarray}}
\newcommand{\nms}{\IEEEeqnarraynumspace}

\newcommand{\bp}{\begin{problem}}
\newcommand{\ep}{\end{problem}}
\newcommand{\hso}{\hspace{.1in}}

\newcommand{\noi}{\noindent}

\newtheorem{problem}{Problem}
\newtheorem{theorem}{Theorem}
\newtheorem{remar}{Remark}

\newtheorem{definition}{Definition}

\begin{document}
%
\title{Capacity of Binary State Symmetric Channel with and without Feedback and Transmission Cost \vspace{-0.5cm}}
\author{\IEEEauthorblockN{Christos K. Kourtellaris}
\IEEEauthorblockA{Department of Electrical and Computer Engineering\\
Texas A{\&}M University at Qatar\\
Email: c.kourtellaris@tamu.edu}
\and
\IEEEauthorblockN{Charalambos D. Charalambous}
\IEEEauthorblockA{Department of Electrical and Computer Engineering\\
University of Cyprus\\
Email: chadcha@ucy.ac.cy}}


%


\maketitle

\begin{abstract}
We consider a unit memory channel, called Binary State Symmetric Channel (BSSC),
in which the channel state is the modulo2 addition of the
current channel input and the previous channel output. We derive closed form expressions for the capacity and  corresponding channel input distribution, of this BSSC with and without feedback and transmission cost. We also show that the capacity of the BSSC is not increased by feedback, and it is achieved by a first order symmetric Markov process. 

\end{abstract}


\section{Introduction}
Capacity of channels with feedback and associated coding theorems are often classified into Discrete Memoryless Channels (DMCs) and channels with memory. For DMCs with and without feedback,  coding theorems, capacity expressions, and achieving distributions are derived by 
Shannon \cite{shannon1956} and Dubrushin \cite{dubrushin1958}. Coding theorems for channels with memory  and feedback for stationary ergodic processes, directed information stable processes, and general nonstationary processes   are given in  \cite{kramer1998,tatikonda2000}. 
\par Although for several years great effort has been devoted to the study of channels with memory, with or without feedback, explicit or closed form expressions of  capacity for channels with memory, are limited to few but ripe cases. Some of these are the additive Gaussian noisy channels with memory and feedback\cite{ebert1970,cover-pombra1989} where the authors proved that memory can increase the capacity of channels with feedback, the first order moving average Gaussian noise channel \cite{kim06}, the trapdoor channel where it was shown the that feedback capacity is the log of the golden ratio \cite{permuter08}, and the Ising channel \cite{elishco}.

\par The capacity of the Unit Memory Channel (UMC) with feedback, defined by $\{P_{B_i|A_i,B_{i-1}}(b_i|a_i,b_{i-1}): i=0,1, \ldots\}$, where $b_i$ is the channel output and $a_i$ the channel input,   is investigated by Berger \cite{berger_shannon_lecture} and  Chen and Berger \cite{chen-berger2005}. Let $a^n\tri\{a_0, a_1, \ldots, a_n\}$ and similarly for $b^n$.  It is conjectured in  \cite{chen-berger2005} that the capacity achieving distribution has the property\footnote{The authors were not able to locate the derivation of (\ref{berger_1}); this property is verified in \cite{kourtellaris2014}.}  
\beae
P_{A_i|A^{i-1}, B^i}(a_i|a^{i-1}, b^{i-1})=P_{A_i|B_i}(a_i|b_{i-1}),  i=0,\ldots \nms \label{berger_1}
\eeae
\par Recently,   Asnani, Permuter and Weissman \cite{asnani13,asnani13j} obtained an expression of the capacity of the so-called Previous Output STate (POST) channel, which is a special case of  (\ref{berger_1}).  They  have shown that feedback does not increase the capacity of the POST channel,  among other results.

\par This paper is concerned with the  Binary State Symmetric channel $(BSSC)$ defined by  (\ref{gench}) with and without feedback and transmission cost.  
Our interest in the $BSSC$ is motivated by the desire to identify a {\it Duality of  Sources and channels}, in the sense of Joint Source-Channel Coding (JSCC) design, in which the optimal transmission is nonanticipative and the nonanticipative Rate Distortion function of a source with memory \cite{ppc2014isit,kourtellaris2014}  is matched to the capacity of a channel with feedback. With respect to this motivation, the  $BSSC$  and Binary Symmetric Markov Source (BSMS) are a natural generalization of   the JSCC design (uncoded transmission) of an Independent and Identically Distributed (IID) Bernoulli source over a Binary Symmetric Channel (BSC)  \cite{gastpar}. The duality of the $BSSC$  and BSMS source is discussed in the companion paper \cite{ksc2014_matching}, and utilizes  the results  of this paper. 
\par The main results derived in this paper are the  explicit expressions of the capacity and corresponding achieving channel input distribution of the $BSSC$, with and without feedback and transmission  cost.

 Since the POST channel \cite{asnani13,asnani13j}  and the BSSC defined by (\ref{gench}) are within a transformation equivalent, our results for the case  without transmission cost, compliment the results in \cite{asnani13,asnani13j}, in the sense that, we give  alternative direct derivations and  we obtain the expression of the capacity achieving channel input distribution with feedback. Moreover, we show that a Markov channel input distribution achieves the capacity of the channel when there is no feedback, hence feedback does not increase capacity of the BSSC.  Our capacity formulae highlights the optimal time sharing among two binary symmetric channels
(states of the general unit memory channel). The case with transmission cost is necessary for the JSCC design found in  \cite{ppc2014,kourtellaris2014}.


\section{Problem Formulation}
\label{cabistsych}
\par  In this section we present the optimization problems which correspond to  the capacity of channels with memory with and without feedback and transmission cost, and  discuss the special classes of UMCs and 
BSSCs. 

Let $\mathbb{N}^n\tri\{0, 1,2,\dots,n\}$, $n\in{\mathbb{N}}\tri\{0, 1,2,\dots\}$, ${\cal A},{\cal B}$ denote the channel input and output alphabets, respectively, and
${\cal A}^n\tri\times_{i=0}^{n}{\cal A}$ and ${\cal B}^n\tri\times_{i=0}^{n}{\cal B}$,
their product spaces, respectively. Moreover, let
$a^n\tri\{a_0, a_1,\dots, a_n\}\in{\cal A}^n$ denote the channel input sequence
of length $n+1$, and similarly for $b^n$.
We associate the above product spaces by their measurable spaces
$({\cal A}^n,\mathbb{B}({\cal A}^n))$, $({\cal B}^n,\mathbb{B}({\cal B}^n))$.

\begin{definition}({Channels with memory with and without feedback})\\
A channel with memory is a sequence of conditional
distributions $\{P_{B_i|B^{i-1},A^{i}}(d{b}_i|b^{i-1},a^{i}):i\in{\mathbb{N}}^n\}$ defined by
\bea
{\overrightarrow P}_{B^n|A^{n}}(d{b}^n|a^{n})\tri
\otimes_{i=0}^{n}P_{B_i|B^{i-1},A^{i}}(d{b}_i|b^{i-1},a^{i}). \nonumber 
\eea
The channel input distribution of a channel with feedback is a sequence of conditional
distributions $\{P_{A_i|A^{i-1},B^{i-1}}(d{a}_i|a^{i-1},b^{i-1}):i\in{\mathbb{N}}^n\}$
defined by
\beae
{\overleftarrow P}_{A^n|B^{n{-}1}}(d{a}^n|b^{n{-}1})\tri
\otimes_{i=0}^{n}P_{A_i|A^{i{-}1},B^{i{-}1}}(d{a}_i|a^{i{-}1},b^{i{-}1}).\nonumber\\ \nms
 \label{defin1id}
\eeae
The channel input distribution of a channel without feedback is a sequence of conditional
distributions $\{P_{A_i|A^{i-1}}(d{a}_i|a^{i-1}):i\in{\mathbb{N}}^n\}$ defined by
\bea
P_{A^n}(a^n)=\otimes_{i=0}^{n}P_{A_i|A^{i-1}}(d{a}_i|a^{i-1}).
\eea
\end{definition}

\begin{definition}(Transmission cost)\\
The cost of transmitting symbols over a channel with memory is 
 a measurable function $\gamma_{0,n}:{\cal A}^{n}\times{\cal B}^{n-1}\mapsto [0,\infty)$  defined by 
\vspace{-0.5cm} 
\bea
\gamma_{0,n}(a^n,b^{n-1})\tri\sum_{i=0}^{n}{c}_{0, i}(a^i,b^{i-1}).
\eea
The transmission cost constraint of a channel with feedback  is defined by 
\beae
{\cal P}_{0,n}^{fb}(\kappa) {\tri} \Big\{ {\overleftarrow P}_{A^n|B^{n-1}}{:} {\frac{1}{n{+}1}} {\bf E}\big\{\gamma_{0,n}(a^n,b^{n{-}1})\big\} \leq \kappa\Big\},\nms
\eeae
where $\kappa \in [0, \infty)$.\\
The transmission cost constraint of a channel without feedback  is defined by
\begin{align}
{\cal P}_{0,n}^{nfb}(\kappa) \tri \Big\{P_{A^n}: \frac{1}{n+1} {\bf E}\big\{\gamma_{0,n}(a^n,b^{n-1})\big\}\leq \kappa\Big\} .
\end{align}
\end{definition}
\noi Define the following quantities. 
\begin{align}
C_{0,n}^{fb}(\kappa) \tri &  \sup_{{\overleftarrow P}_{A^n|B^{n-1}}  \in {\cal P}_{0,n}^{fb}(\kappa)  }{\frac{1}{n+1}}I(A^n\rar B^n),  \\
C_{0,n}^{nfb}(\kappa) \tri & \sup_{P_{A^n}  \in {\cal P}_{0,n}^{nfb}(\kappa)  }{\frac{1}{n+1}}I(A^n; B^n).
\end{align}
where $I(A^n \rar B^n) \tri \sum_{i=0}^n I(A^i; B_i|B^{i-1})$. If there is no transmission cost the above expressions are denoted by $C_{0,n}^{fb}, C_{0,n}^{nfb}$. \\
\par Note that for a fixed channel distribution, ${\overrightarrow P}_{B^n|A^{n}}(d{b}^n|a^{n})$, the set of causal conditional distributions  ${\overleftarrow P}_{A^n|B^{n-1}}$ is convex, which implies  that 
 $I(A^n\rar B^n)$ is a convex functional of ${\overleftarrow P}_{A^n|B^{n-1}}$, and that the transmission cost constraint ${\cal P}_{0,n}^{fb}(\kappa)$ is a  convex  set (see \cite{charalambous-stavrou2012}). Hence,   $C_{0,n}^{fb}(\kappa)$ is a convex optimization problem. The fact that $C_{0,n}^{nfb}(\kappa)$ is a convex optimization problem is well known.  
\par Under the assumption that $\{(A_i, B_i): i=1, 2, \ldots, \}$ is jointly ergodic or $\frac{1}{n+1}\log \frac{{\overrightarrow P}_{B^n|A^{n}}(d{b}^n|a^{n})}{ P_{B^n}(d{b}^n)}$ is information stable \cite{dubrushin1958,tatikonda2000}, then the channel capacity with feedback encoding and without feedback encoding are given by 
\begin{align}
C^{fb} \tri \lim_{n \longrightarrow  \infty} C_{0,n}^{fb}, \hso C^{nfb} \tri  \lim_{n \longrightarrow \infty} C_{0,n}^{nfb}.
\end{align}
Under appropriate assumptions   then $C^{fb}(\kappa)$ and  $C^{nfb}(\kappa)$ corresponds to the channel capacity as well.

\subsection{Unit Memory Channel with Feedback}
A Unit Memory Channel (UMC) is defined by
\bea
{\overrightarrow P}_{B^n|A^{n}}(d{b}^n|a^{n})\tri
\otimes_{i=0}^{n}P_{B_i|B_{i},A_{i-1}}(d{b}_i|b_i,a_{i-1}).\label{umc}
\eea
\par For a UMC the following results are found in \cite{chen-berger2005}. If the channel is indecomposable  then 
\beae
C^{fb} &{=}&\lim_{n {\longrightarrow} \infty}\sup_{{\overleftarrow P}_{A^n|B^{n{-}1}}}{\frac{1}{n{+}1}}I(A^n{\rar} B^n)\nonumber\\
&{=}&{\lim_{n\longrightarrow \infty}}{\sup_{\{P(A_i|B_{i-1})\}_{i=0}^{n}}}\frac{1}{n+1}\sum_{i=0}^{n}I(A_i; B_i|B_{i-1})\nms\label{chen-burger_1} \\
&{=}&\sup_{P(A_i|B_{i{-}1})}\big\{H(B_i|B_{i{-}1}){-}H(B_{i}|A_i,B_{{i}{-}1})\big\}, \; \forall  i \nms
\label{chen-burger_2} 
\eeae
and the following hold.
\begin{enumerate}
\item $\{(A_i,B_i):i=0, 1,2,\ldots,\}$ is a first order stationary Markov process.
\item $\{B_i:i=0, 1,2,\ldots,\}$ is a first order stationary Markov process.
\end{enumerate}
\par Suppose the cost of  transmitting symbols is letter-by-letter  and time invariant, restricted to  $\gamma_{0,n}(a^n,b^{n-1})= \sum_{i=0}^{n}{c}(a_i,b_{i-1})$. Then by repeating 
the derivation in \cite{chen-berger2005}, if necessary, by introducing the Lagrangian functional associated with the average transmission cost constraint (and assuming existence of an interior point of the constraint), then 
 \medmuskip=-1mu
\thinmuskip=-1.6mu
\thickmuskip=-1mu
\beae
&&C^{fb}(\kappa)=\lim_{n\longrightarrow \infty}\sup_{ \{P_{A_i|B_{i-1}}\}_{i=0}^n:\frac{1}{n+1}\sum_{i=0}^{n}{\bf E}\big\{c(A_i,B_{i-1})\big\}\leq \kappa} \nonumber \\&& \hspace{1.4cm}\frac{1}{n+1} \Big\{ 
 \sum_{i=0}^n I(A_i; B_i|B_{i-1})\Big\}\nonumber \\
&=&\sup_{P_{A_i|B_{i{-}1}}:{\bf E}\big\{ c(A_i,B_{i{-}1})\big\}\leq \kappa}\Big\{H(B_i|B_{i{-}1}){-}H(B_{i}|A_i,B_{{i}{-}1})\Big\},\nonumber \\  \label{christos_1}  
\eeae
where for $\kappa \in  [0, \kappa_m]$ and $\kappa_m$ the maximum value of the cost constraint, both 1) and 2) remain valid. Moreover, $C^{fb}(\kappa)$ 
 is a concave, nondecreasing function of $\kappa \in [0, \kappa_m]$. A derivation of the fact that for a UMC with transmission cost $\gamma_{0,n}(a^n,b^{n-1})= \sum_{i=0}^{n}{c}(a_i,b_{i-1})$ the optimal channel input distribution has the conditional independence property $P_{A_i|A^{i-1}, B^{i-1}}=  P_{A_i|B_{i-1}}, \ i=0,1, \ldots,$   is also given in \cite{kourtellaris2014}.  This is necessary to obtain (\ref{chen-burger_2})-(\ref{christos_1}), and subsequently to show properties 1) and 2),   

Unfortunately, the closed form expression of the capacity achieving channel input distribution of the UMC (even when the input and output alphabets are binary)  is currently unknown. 

\subsection{The Binary State Symmetric Channel} 
\par In this section, we consider a special class  of the UMC channel, the BSSC, and discuss its physical meaning, and that  of the imposed transmission cost constraint. 
\par The $BSSC({\alpha}, {\beta})$ is a  unit memory
channel defined by 
 \medmuskip=-1mu
\thinmuskip=-1.6mu
\thickmuskip=-1mu
\beae
 P_{B_i|A_i,B_{i{-}1}}(b_i|a_i,b_{i{-}1}) {=} \bbordermatrix{~ & 0,0 &\hspace{-0.22cm} 0,1 &\hspace{-0.22cm} 1,0 &\hspace{-0.22cm} 1,1 \cr
                  0 & \alpha &\hspace{-0.22cm} \beta &\hspace{-0.22cm} 1{-}\beta &\hspace{-0.22cm} 1{-}\alpha \vspace{0.3cm}\cr 
                  1 & 1{-}\alpha &\hspace{-0.22cm} 1{-}\beta  &\hspace{-0.22cm} \beta &\hspace{-0.22cm} \alpha \cr}.\nms
                  \label{gench}
\eeae
\par Introduce the following change of variables called state of the channel, $s_i \tri a_i\oplus b_{i-1}, \ i\in{\mathbb N}^n$, where $\oplus$ denotes
the modulo2 addition. 
This transformation is one to one and onto, in the sense that,  for a fixed  channel input symbol value $a_i$ (respectively channel output symbol value $b_{i-1})$ then $s_i$ is uniquely determined  by the value of $b_{i-1}$ (respectively $a_i$) and vice-verse. 
Then the following equivalent representation of the BSSC is obtained.
\bea
P_{B_i|A_i,S_i}(b_i|a_i,s_i=0)  &=& \bbordermatrix{~ \cr
                  & \alpha & 1-\alpha \cr
                  & 1-\alpha & \alpha \cr}.   \label{state_zero}  \\
                  P_{B_i|A_i,S_i}(b_i|a_i,s_i=1)&=& \bbordermatrix{~ \cr
                  & \beta & 1-\beta \cr
                  & 1-\beta & \beta \cr}. \label{state_one}
                  \eea
\par The above transformation  highlights the symmetric form of the BSSC
for a fixed  state $s_i\in\{0,1\}$, which  decomposes (\ref{gench})  into  binary symmetric channels with crossover probabilities $(1-\alpha)$ and $(1-\beta)$, and motivates the name state symmetric channel.
\par The following notation will be used in the rest of the paper. 
\begin{itemize}
\item[1)] $BSSC(\alpha,\beta)$ denotes the BSSC with transition probabilities defined by (\ref{gench});
\item[2)] $BSC(1-\alpha)$ denotes the ``state zero" channel defined by (\ref{state_zero});
\item[3)] $BSC(1-\beta)$ denotes the ``state one" channel defined by (\ref{state_one}).
\end{itemize}

Next, we discuss the physical interpretation of the cost constraint. Consider $\alpha>\beta\geq 0.5$. 
Then the capacity of the state zero channel, $(1-H(\alpha))$, is greater than the capacity of the state one channel, $(1-H(\beta))$. With ``abuse" of terminology, we interpret the state zero channel as the ``good channel" and the state one channel, as the ``bad channel". It is then reasonable to consider a higher cost 
when employing the ``good channel" and a lower cost when employing the ``bad channel".  We quantify
this policy by assigning  the following binary pay-off to each of the channels.
\begin{equation}
c_i(a_i,b_{i-1}) = \left\{
  \begin{array}{l l}
    1 & \quad \text{if $a_i=b_{i-1}$, $(s_i=0)$}\\
    0 & \quad \text{if $a_i\neq b_{i-1}$, $(s_i=1)$ }
 \label{sds} \end{array}  \right. 
 \end{equation}
The letter-by-letter average transmission cost is given by
\beae
{\bf E}\{c(A_i,B_{i-1})\}
=P_{A_i,B_{i{-}1}}(0,0)+P_{A_i,B_{i{-}1}}(1,1)=P_{S_i}(0)
. \hspace{-0.15cm}\nms \label{qvcostc1} 
\eeae

\begin{remar} The binary form of the constraint does not downgrade the problem, since it can
be easily upgraded to more complex forms, without affecting the proposed methodology
(i.e. $(1-\delta)$, $\delta$, where $\delta=constant$).
Moreover, for $\beta>\alpha\geq 0.5$ we reverse the cost, while for $\alpha$ and/or $\beta$ are less
than $0.5$ we flip the respective channel input.
\end{remar}


\section{Explicit Expressions of Capacity of  BSSC with Feedback with \& without  Transmission Cost}\label{gencap}

\begin{figure}
\vspace{-0.4cm}
\includegraphics[scale=0.17]{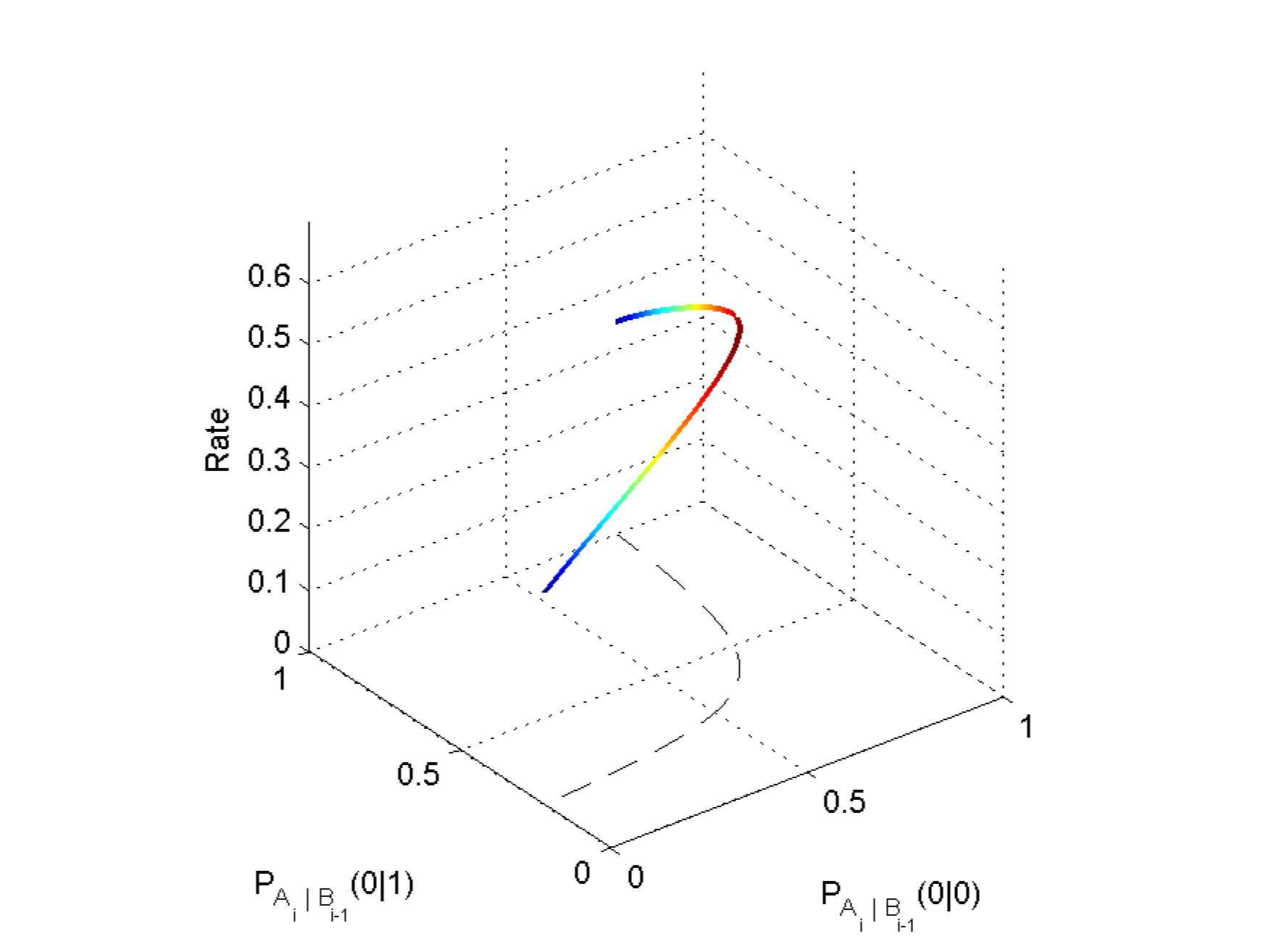}
\caption{Rate of the BSSC with feedback subject to transmission cost constraint (with equality) for ${\alpha}=0.92$, ${\beta}=0.79$ and $\kappa=0.71$.}
\label{con_rate}
\end{figure}

\par In this section we provide explicit (or closed form) expressions for the capacity of the BSSC with feedback,  with and without transmission cost.
\vspace{-0.3cm}  
\subsection{Capacity with Feedback and Tranmission Cost}
\vspace{-0.1cm}
\par Consider the case when there is feedback and transmission cost. Without loss of generality in the optimization problem (\ref{christos_1}) we replace the inequality  by an equality, because the optimization problem is convex, and hence the optimal channel input distributions occurs on the boundary of the constraint, provided $\kappa \in [0, \kappa_m]$, where $\kappa_m \in [0,1]$. 
\vspace{-0.1cm}
\par Hence, we discuss the problem with an equality transmission cost constraint defined by 
\beae
C^{fb}(\kappa)=\sup_{P_{A_i|B_{i-1}}:{\bf E}\{{c}(A_i,B_{i-1})\}=\kappa}I(A_i; B_i|B_{i-1}).
\nms\label{ccende}
\eeae
where $\kappa \in [0,1]$. In section \ref{sec_feed_nc} 
(Remark~\ref{equ_eq_ineq}), we  discuss the case when inequality is considered. 
\vspace{-0.1cm}
\par The constraint rate of the BSSC with feedback
is illustrated in Figure~\ref{con_rate}. The projection on the distribution plane,
denoted by the black dotted line, shows all possible pairs of input distributions that satisfy the transmission cost ${\bf E}\{{c}(A_i,B_{i-1})\}=\kappa$. 
\vspace{-0.1cm}
\par Next, we state the main theorem from which all other results  (no feedback, inequality transmission cost, no transmission cost) will be derived. 

\begin{theorem}(Capacity of $BSSC(\alpha,\beta)$ with feedback \& transmission cost)\\
\label{cftc}
The capacity of $BSSC(\alpha,\beta)$ with feedback and transmission cost ${\mathbb E}\{{c(A_i,B_{i-1})}\}=\kappa, \ \kappa \in [0,1]$ is
given by
\begin{IEEEeqnarray}{l}
C^{fb}(\kappa)=H(\lambda){-}\kappa H({\alpha}){-}(1{-}\kappa)H({\beta}).\IEEEeqnarraynumspace
 \end{IEEEeqnarray}
where $\lambda=\alpha\kappa+(1-\kappa)(1-\beta)$. \\
The optimal input and output distributions are given by
\vspace{-0.3cm}
 \medmuskip=-1mu
\thinmuskip=-1.6mu
\thickmuskip=-1mu
\bea
P^{*}_{A_i|B_{i-1}}(a_i|b_{i-1}) = \bbordermatrix{~ \cr
                  & \kappa & 1-\kappa \cr
                   & 1-\kappa & \kappa \cr},\vspace{-0.8cm}
\label{ccon_inp_dis}\\ 
P^{*}_{B_i|B_{i-1}}(b_i|b_{i-1}) = \bbordermatrix{~\cr
                  & \lambda & 1-\lambda \cr
                   & 1-\lambda & \lambda \cr}.
\label{ccon_out_dis}
\eea
\end{theorem}

\begin{figure}
\begin{center}
\includegraphics[scale=0.23]{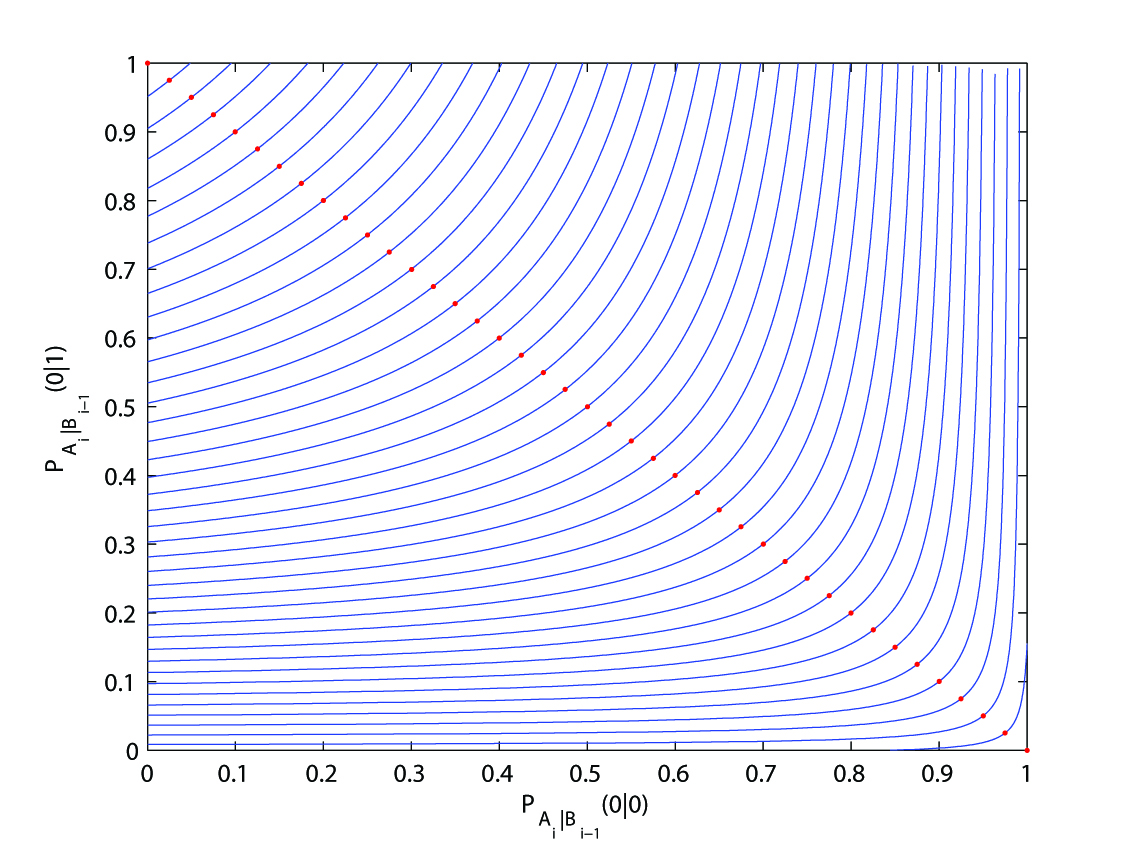}
\caption{Input distributions and the respective capacity for $\kappa=0,0.025,0.05\ldots$,1}
\label{cap_k}
\end{center}
\end{figure}

\begin{proof} We outline the proof. The second term of the RHS of (\ref{ccende}) is
fixed by the cost constraint, and it is given by
\bea
H(B_i|B_{i-1},A_i)=\kappa H({\alpha})+(1-\kappa)H({\beta}). \label{cha_ent}
\eea
The conditional distribution of the output is given by
\bea
P_{B_i|B_{i-1}}=\sum_{a_i \in {\cal A}}P_{B_i|A_i,B_{i-1}}P_{A_i|B_{i-1}}.\label{reoudix}
\eea
Then, by manipulating (\ref{qvcostc1}) and (\ref{reoudix}), we obtain
\beae P_{B_i|B_{i-1}}(0|0)P_{B_i}(0){+}P_{B_i|B_{i{-}1}}(1|1)(1{-}P_{B_i}(0))=\lambda, \nms \label{ch4exp1}
\eeae
where $\lambda=\alpha\kappa+(1-\beta)(1-\kappa)$.
Using (\ref{ch4exp1}) we obtain the following expressions for $P_{B_i|B_{i-1}}(0|0)$
and $P_{B_i}(0)$, as functions of $P_{B_i|B_{i-1}}(1|1)$, $\alpha,\beta,\kappa$.
\beae
P_{B_i}(0)&=&\frac{1+\lambda-2P_{B_i|B_{i-1}}(1|1)}{2(1-P_{B_i|B_{i-1}}(1|1))},\label{ch4exp3}\\
P_{B_i|B_{i-1}}(0|0)&=&\frac{2\lambda-(1+\lambda)P_{B_i|B_{i-1}}(1|1)}{1+\lambda-2P_{B_i|B_{i-1}}(1|1)}. \label{ch4exp4}
\eeae
To simplify the notation, we set $q_b\tri P_{B_i|B_{i-1}}(1|1)$, and then calculate $H(B_i|B_{i-1})$ as a function of $\lambda$ and $q_b$. Maximizing $H(B_i|B_{i-1})$ with respect to $q_b$, yields
\bea
&& \frac{1-\lambda}{2(q_b-1)^2}
\left(\log\left(\frac{2\lambda-(1+\lambda)q_b}{1+\lambda-2q_b}\right)-\log q_b\right)=0\nonumber\\
&\Rightarrow& \frac{1-\lambda}{2(q_b-1)^2}\log\left(\frac{2\lambda-(1+\lambda)q_b}{(1+\lambda-2q_b)q_b}\right)=0,
\eea
hence,  $q_b=\lambda$ or  
$q_b=1$ (the trivial solution). 
By substituting the non-trivial solution $q_b=\lambda$ into the single letter expression of the constraint capacity we obtain (\ref{ccon_inp_dis}), (\ref{ccon_out_dis}). 
Moreover, since the transition matrix (\ref{ccon_out_dis}) is doubly stochastic, the distribution of the output symbol, $B_i$, at each time instant $i\in\mathbb{N}$, is given by $P_{B_i}(0)=P_{B_i}(1)=0.5, \ i\in\mathbb{N}$.\\

\end{proof}

\par The possible pairs of input distributions (blue curves) and the 
pairs of input distributions that achieve the capacity (red points), for various values of $\kappa$, are illustrated in Figure~\ref{cap_k}.
The input distribution that achieves the capacity satisfies the equation $P_{A_i|B_{i-1}}(0|0)+P_{A_i|B_{i-1}}(0|1)=1$. This can be shown  by substituting
(\ref{ccon_out_dis}) and $P_{B_i}(0)=P_{B_i}(1)=0.5$ in
(\ref{ch4exp1}).

\subsection{Capacity with Feedback without Transmission Cost}\label{sec_feed_nc}
\par When there is no transmission cost constraint   any channel input distribution pair is permissible. An example for a possible rate of the BSSC with feedback without transmission cost is illustrated in Figure~\ref{uncon_rate}. 
\par Next, we derive the analogue of Theorem~\ref{cftc}, when there is no transmission cost.  For this case, the expression of the capacity  highlights the optimal time sharing between the two states.

\begin{theorem}
(Capacity of $BSSC(\alpha,\beta)$ with feedback without transmision cost)\\
\label{cfntc}
The capacity of the $BSSC(\alpha,\beta)$ with feedback without transmission cost is
given by
\begin{IEEEeqnarray}{l}
C^{fb}=H(\lambda^*){-}\kappa^* H({\alpha}){-}(1{-}\kappa^*)H({\beta}),\label{capunfinal}
\end{IEEEeqnarray}
where 
\beae
\lambda^*=&\alpha\kappa^*+(1-\kappa^*)(1-\beta), \\
{\kappa}^*=&\frac{{\beta}(1+2^{\frac{H({\beta})-H({\alpha})}{{\alpha}+{\beta}-1}})-1}
{{({\alpha}+{\beta}-1)}(1+2^{\frac{H({\beta})-H({\alpha})}{{\alpha}+{\beta}-1}})}.\label{oiduc}
\eeae

\begin{figure}
\begin{center}
\vspace{-0.4cm}
\includegraphics[scale=0.125]{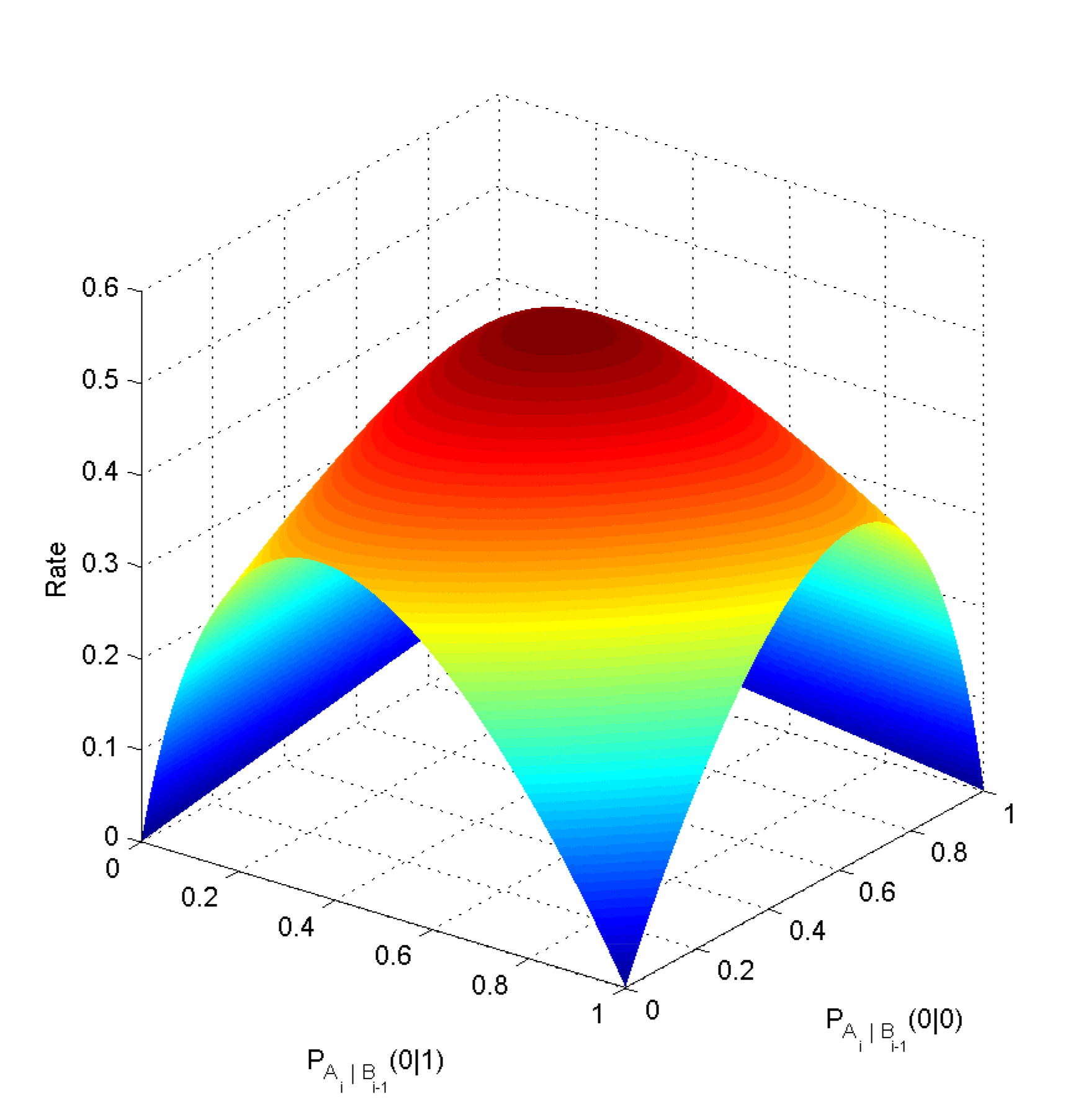}
\caption{Rate of the BSSC with feedback, for $ {\alpha}=0.92$ and ${\beta}=0.79$.}
\label{uncon_rate}
\end{center}
\end{figure}

The optimal input and output distributions are given by
\vspace{-0.5cm}
\bea
P^{*}_{A_i|B_{i-1}}(a_i|b_{i-1}) = \bbordermatrix{~  \cr
                   & \kappa^* & 1-\kappa^* \cr
                   & 1-\kappa^* & \kappa^* \cr}.
\label{con_inp_dis}
\eea
\vspace{-0.5cm}
\bea
P^{*}_{B_i|B_{i-1}}(b_i|b_{i-1}) = \bbordermatrix{~  \cr
                   & \lambda^* & 1-\lambda^* \cr
                   & 1-\lambda^* & \lambda^* \cr}.
\label{con_out_dis}
\eea
\end{theorem}
\begin{proof}
The derivation of the results can be done by following that of Theorem~\ref{cftc}, with   the maximization of directed information taken 
over all possible channel input distributions. Alternatively, by utilizing the statements of Theorem~\ref{cftc}, the capacity without transmission cost corresponds to   the double maximization over the channel input
distributions that satisfy the average cost constraint, and over all
possible values of  $\kappa \in [0,1]$, via
\bea
C^{fb}=\max_{\kappa\in[0,1]}C^{fb}(\kappa)\label{uncon_capac}
\eea
Using (\ref{uncon_capac}) then  (\ref{oiduc}) is obtained; the rest of the statements are easily shown. 
\end{proof}

\par The result of the unconstraint capacity with feedback is equivalent to \cite{asnani13j}. However, since our derivations are different,  the capacity formulae given here highlights the optimal time sharing, $k^*$, among the two binary symmetric channels.
\begin{remar}\label{equ_eq_ineq}
For problem (\ref{christos_1}) with inequality constraint, in view of its convexity, and the fact that $C^{fb}(\kappa)$ as a function of $\kappa$ is  concave and  nondecreasing in $\kappa \in [0, \kappa_m]$, then $\kappa_m ={\kappa^*}$, and  the solution occurs on the boundary of the cost constraint. 
\end{remar}

\section{Explicit Expressions of Capacity of BSSC without Feedback with \& without Transmission Cost}\label{ch4bsscnf}
In this section we show that for BSSC feedback does not increase capacity, and then we derive the analogue of Theorem~\ref{cftc} and Theorem~\ref{cfntc}. 

\begin{theorem}\label{ch4spnf}
{\bf (a):} For the BSSC$(\alpha,\beta)$ with transmission cost,   the first-order Markovian input distribution
$\{P^*_{A_i|A_{i-1}}: i=0,1, \ldots\}$ given by
\bea
P^{*}_{A_i|A_{i-1}}(a_i|a_{i-1}) = \bbordermatrix{~ &  &  \cr
                   & \dfrac{1-\kappa-\gamma}{1-2\gamma} & \dfrac{\kappa-\gamma}{1-2\gamma}   \cr
                   & \dfrac{\kappa-\gamma}{1-2\gamma}   & \dfrac{1-\kappa-\gamma}{1-2\gamma} \cr},
                  \label{opima} \\ \nonumber
\eea
where $\gamma={\alpha}{\kappa}+{\beta}({1-\kappa})$, 
induces the optimal channel input and channel output distributions $P^*_{A_i|B_{i-1}}$ and $P^*_{B_i|B_{i-1}}, P^{*}_{B_{i-1}}$, respectively,  of the BSSC$(\alpha,\beta)$ with feedback and transmission cost.\\
{\bf (b):} For the BSSC$(\alpha,\beta)$ without transmission cost {\bf (a)} holds with  $\kappa=\kappa^*$ and $\gamma=\gamma^*$.\\
{\bf (c):}  The capacity the BSSC without feedback and transmission cost is given by
\beae
C^{nfb}= \sup_{{P}_{A_i|A_{i-1}}}I(A_i;B_i|B_{i-1})= C^{fb},
\eeae
and similarly, if there is a transmission cost.

\end{theorem}
\begin{proof} To prove the claims it suffices to show that a Markovian input distribution
achieves the capacity achieving channel input distribution with feedback. Consider the following identities.
\beae
P^*_{A_i|B_{i{-}1}}&=&\sum_{A_{i{-}1}}\hspace{-0.1cm}P_{A_{i}|A_{i{-}1},B_{i{-}1}}P_{A_{i{-}1}|{B_{i{-}1}}}\nonumber\\
&\sr{?}{=}&\sum_{A_{i{-}1}}\hspace{-0.1cm}P_{A_{i}|A_{i{-}1}}P_{A_{i{-}1}|{B_{i{-}1}}}\nonumber\\
&=&\sum_{A_{i{-}1}}\hspace{-0.15cm}\frac{P_{A_{i}|A_{i{-}1}}}{P_{B_{i{-}1}}}\sum_{B_{i{-}2}}\hspace{-0.12cm}{P_{B_{i{-}1}|A_{i{-}1},B_{i{-}2}}}
{P_{A_{i{-}1}|B_{i{-}2}}P_{B_{i{-}2}}}.\nonumber\\ \nms \label{idnfq}
\eeae
Thus, we search for an input distribution without feedback $P_{A_{i}|A_{i{-}1},B_{i{-}1}}=P_{A_{i}|A_{i{-}1}}$ that satisfies (\ref{idnfq}). Solving iteratively this system of equations yields the values of the optimal input distribution without feedback given by (\ref{opima}). Since $P^*_{A_i|A_{i-1}}$ given by (\ref{opima}) induces $P^*_{A_i|B_{i-1}}$, then the input distribution without feedback also induces the optimal output distribution $P^*_{B_i|B_{i-1}}=\sum_{A_i}P^*_{B_i|B_{i-1},A_i}P^*_{A_i|B_{i-1}}$ and joint processes $P^*_{A_i, B_i|A_{i-1},B_{i-1}}=P^*_{B_i|B_{i-1},A_i}P^*_{A_i|B_{i-1}}$. This is sufficient to conclude {\bf (c)}. 
\end{proof}

\section{Conclussions}
In this paper we formulate the capacity of the UMC and the BSSC and provide the  explicit expressions of the capacity and corresponding achieving channel input distribution of the $BSSC$, with and without feedback and with and without transmission  cost. 

%

\vspace{-0.1cm}
\bibliographystyle{IEEEtran}
\bibliography{Bibliography}
\end{document}